\documentclass[runningheads]{llncs}
\pdfoutput=1
\usepackage[letterpaper=true,pdfview=FitH,pdfstartview=FitH,bookmarks=false]{hyperref}
\usepackage[T1]{fontenc}
\usepackage[latin9]{inputenc}
\usepackage{float}
\usepackage{amsmath,amsfonts,amssymb,url,verbatim,stmaryrd,mdwlist}
\usepackage{graphics,url}
\usepackage{epsfig, subfigure}
\usepackage{amsfonts}
\usepackage{multirow}
\usepackage{hhline}
\newcommand{\term}[1]{\textit{#1}}

\newcommand{\brak}[1]{\ensuremath{[\![#1]\!]}}
\renewcommand{\paragraph}[1]{\medskip\noindent\textbf{{#1}}}
\floatstyle{ruled}
\newfloat{algorithm}{tbp}{loa}
\floatname{algorithm}{Algorithm}

\numberwithin{equation}{section} %% Comment out for sequentially-numbered
\numberwithin{figure}{section} %% Comment out for sequentially-numbered
\newtheorem{thm}{Theorem}
  \newtheorem{rem}[thm]{Remark}
  \newtheorem{lem}[thm]{Lemma}
\newcommand{\indic}[1]{\mathbf{1}\left[#1\right]}
\newcommand{\nth}[2]{${#1}^{\mbox{\scriptsize {#2}}}$}

\makeatother

\DeclareMathOperator{\payoff}{payoff}
\DeclareMathOperator{\reward}{reward}
\DeclareMathOperator{\cost}{cost}
\DeclareMathOperator{\profit}{profit}
\DeclareMathOperator{\ROA}{ROA}

\DeclareMathOperator*{\argmin}{argmin}
\DeclareMathOperator*{\argmax}{argmax}
\DeclareMathOperator{\inc}{inc}
\newcommand{\instmark}[1]{$^{\mbox{\scriptsize #1}}$}

\begin{document}

\mainmatter  % start of an individual contribution

\title{A Learning-Based Approach to Reactive Security}

\titlerunning{A Learning-Based Approach to Reactive Security}

\author{Adam Barth\instmark{1} \and Benjamin~I.~P. Rubinstein\instmark{1} \and Mukund Sundararajan\instmark{3} \and\\
John~C. Mitchell\instmark{4} \and Dawn Song\instmark{1} \and Peter~L. Bartlett\instmark{1,2}}

\authorrunning{Barth, Rubinstein, Sundararajan, Mitchell, Song, Bartlett}

\institute{\instmark{1}Computer Science Division \instmark{2}Department of Statistics, UC Berkeley\\
\instmark{3} Google Inc., Mountain View, CA\\
\instmark{4}Department of Computer Science, Stanford University}

% From Ben: if we later want emails we can first uncomment the urldefs in the
% header, then add the following into the institute command:
%\url{http://www.springer.com/lncs}


\toctitle{Lecture Notes in Computer Science}
\tocauthor{Authors' Instructions}
\maketitle

\begin{abstract}
Despite the conventional wisdom that proactive security is superior to
reactive security, we show that reactive security can be competitive
with proactive security as long as the reactive defender learns from past
attacks instead of myopically overreacting to the last attack. Our
game-theoretic model follows common practice in the security literature
by making worst-case assumptions about the attacker: we grant the attacker
complete knowledge of the defender's strategy and do not require the
attacker to act rationally. In this model, we bound the competitive ratio
between a reactive defense algorithm (which is inspired by online learning
theory) and the best fixed proactive defense. Additionally, we show that,
unlike proactive defenses, this reactive strategy is robust to a lack of
information about the attacker's incentives and knowledge.
\end{abstract}

%
%\maketitle

\section{Introduction}

Many enterprises employ a Chief Information Security Officer~(CISO) to
manage the enterprise's information security risks. Typically, an enterprise
has many more security vulnerabilities than it can realistically repair.
Instead of declaring the enterprise ``insecure'' until every last
vulnerability is plugged, CISOs typically perform a cost-benefit analysis to
identify which risks to address, but what constitutes an effective CISO
strategy?  The conventional wisdom~\cite{pironti-2005,kark-penn-dill-2008}
is that CISOs ought to adopt a ``forward-looking'' proactive approach to
mitigating security risk by examining the enterprise for vulnerabilities that
might be exploited in the future. Advocates of proactive security often
equate reactive security with myopic bug-chasing and consider it ineffective.
We establish sufficient conditions for when reacting \emph{strategically} to
attacks is as effective in discouraging attackers.
 
We study the efficacy of reactive strategies in an economic
model of the CISO's security cost-benefit trade-offs.
Unlike previously proposed economic models of security (see
Section~\ref{sec:related-work}), we do not assume the attacker acts
according to a fixed probability distribution. Instead, we consider
a game-theoretic model with a strategic attacker who responds to the
defender's strategy. As is standard in the security literature, we
make worst-case assumptions about the attacker.  For example, we grant
the attacker complete knowledge of the defender's strategy and do not
require the attacker to act rationally. Further, we make conservative
assumptions about the reactive defender's knowledge and do not assume
the defender knows all the vulnerabilities in the system or the
attacker's incentives. However, we do assume that the defender can
observe the attacker's past actions, for example via an intrusion
detection system or user metrics~\cite{TestPilot}.

In our model, we find that two properties are sufficient for a reactive
strategy to perform as well as the best proactive strategies.  First,
no single attack is catastrophic, meaning the defender can
survive a number of attacks.  This is consistent with situations
where intrusions (that, say, steal credit card numbers) are regrettable but
not business-ending. Second, the defender's budget is
\term{liquid}, meaning the defender can re-allocate resources without
penalty. For example, a CISO can reassign members of
the security team from managing firewall rules to improving database access
controls at relatively low switching costs. 

Because our model abstracts many vulnerabilities into a single graph edge,
we view the act of defense as increasing the attacker's \term{cost} for
mounting an attack instead of preventing the attack (e.g., by patching a
single bug).  By making this assumption, we choose not to study the tactical
patch-by-patch interaction of the attacker and defender. Instead, we model
enterprise security at a more abstract level appropriate for the CISO. For
example, the CISO might allocate a portion of his or her budget to engage a
consultancy, such as WhiteHat or iSEC Partners, to find and fix cross-site
scripting in a particular web application or to require that employees use
SecurID tokens during authentication.
We make the technical assumption that attacker costs are linearly dependent
on defense investments locally.  This assumption does not reflect
patch-by-patch interaction, which would be better represented by a step
function (with the step placed at the cost to deploy the patch).  Instead,
this assumption reflects the CISO's higher-level viewpoint where the
staircase of summed step functions fades into a slope.

We evaluate the defender's strategy by measuring the attacker's cumulative
return-on-investment, the \term{return-on-attack}~(ROA), which has been
proposed previously~\cite{ROA}. By studying this metric, we focus on
defenders who seek to ``cut off the attacker's oxygen,'' that is to reduce
the attacker's incentives for attacking the enterprise.  We do not
distinguish between ``successful'' and ``unsuccessful'' attacks. Instead, we
compare the payoff the attacker receives from his or her nefarious deeds
with the cost of performing said deeds.  We imagine that sufficiently
disincentivized attackers will seek alternatives, such as attacking a
different organization or starting a legitimate business.

% We model \term{reactive} defense as backward looking but agile: a reactive
% defender knows only the history of prior attacks but is able to revise his
% or her defense allocation in light of attacks.  By way of contrast, we model
% an idealized \term{proactive} defense as forward looking but static: a
% proactive defender knows the future history of all attacks but must commit
% to a single defense allocation.  Of course, no CISO knows the future history
% of all attacks, but assumption is an idealization of the best the CISO could
% hope to learn from studying vulnerabilities in the enterprise (e.g., using
% RedSeal or Skybox).

In our main result, we show sufficient conditions for a learning-based
reactive strategy to be
competitive with the best fixed proactive defense in the sense that the
competitive ratio between the reactive ROA and the proactive ROA is at most
$1+\epsilon$, for all $\epsilon > 0$, provided the game lasts sufficiently
many rounds (at least $\Omega(1/\epsilon)$). To prove our theorems, we draw
on techniques from the online learning literature. We extend these
techniques to the case where the learner does not know all the game matrix
rows \emph{a priori}, letting us analyze situations where the defender does
not know all the vulnerabilities in advance.  Although our main
results are in a graph-based model with a single attacker, our results
generalize to a model based on Horn clauses with multiple attackers.  Our
results are also robust to switching from ROA to attacker profit and to
allowing the proactive defender to revise the defense allocation a fixed number
of times.

Although myopic bug chasing is most likely an ineffective reactive strategy, we
find that in some situations a \emph{strategic} reactive strategy is as
effective as the
optimal fixed proactive defense.  In fact, we find that the natural strategy of
gradually reinforcing attacked edges by shifting budget from unattacked
edges ``learns'' the attacker's incentives and constructs an effective
defense.  Such a strategic reactive strategy is both easier to implement
than a proactive strategy---because it does not presume that the defender
knows the attacker's intent and capabilities---and is less wasteful than a
proactive strategy because the defender does not expend budget on attacks
that do not actually occur. Based on our results, we encourage CISOs to 
question the assumption that proactive risk management is inherently
superior to reactive risk management.

\paragraph{Organization.}
Section~\ref{sec:formal-model} formalizes our model.
Section~\ref{sec:case-studies} shows that
perimeter defense and defense-in-depth arise naturally in our model.
Section~\ref{sec:reactive-security} presents our main results bounding the
competitive ratio of reactive versus proactive defense strategies.
Section~\ref{sec:advantages} outlines scenarios in which reactive security
out-performs proactive security.  Section~\ref{sec:generalizations}
generalizes our results to Horn clauses and multiple attackers.
Section~\ref{sec:related-work} relates related work.
Section~\ref{sec:conclusions} concludes.

\begin{figure}[t]
\centering
\includegraphics[width=0.5\columnwidth]{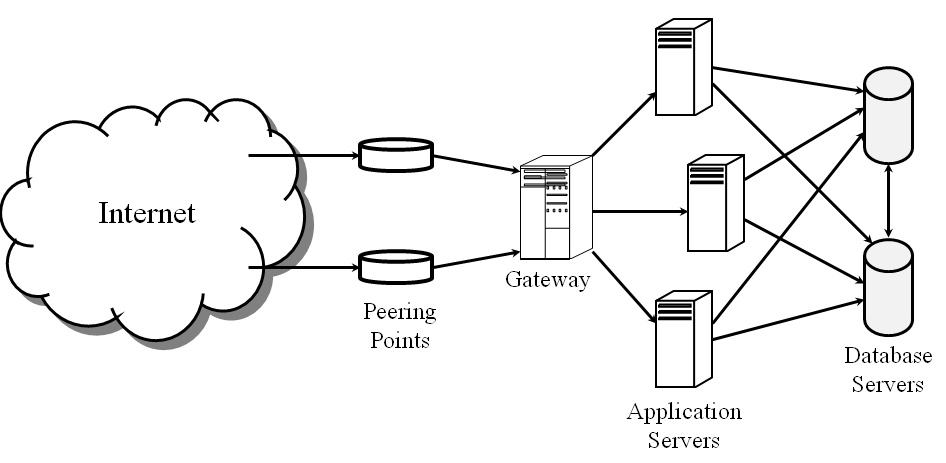}
\caption{An attack graph representing an enterprise data center.}
\label{fig:cloudy}
\end{figure}

%\presec

\section{Formal Model} \label{sec:formal-model}

In this section, we present a game-theoretic model of attack and defense.
Unlike traditional bug-level attack graphs, our model is meant to capture
a managerial perspective on enterprise security.  The model is somewhat
general in the sense that attack graphs
can represent a number of concrete situations, including a network (see
Figure~\ref{fig:cloudy}), components in a complex software
system~\cite{multiprocessarch}, or an Internet Fraud
``Battlefield''~\cite{battlefield}.

\paragraph{System.}
We model a system using a directed graph $(V, E)$, which defines the game
between an attacker and a defender.  Each vertex $v \in V$ in the graph
represents a state of the system.  Each edge $e \in E$ represents
a state transition the attacker can induce.  For example, a vertex might
represent whether a particular machine in a network has been compromised by
an attacker.  An edge from one machine to another might represent that an
attacker who has compromised the first machine might be able to compromise
the second machine because the two are connected by a network.
Alternatively, the vertices might represent different components in a
software system.  An edge might represent that an attacker sending input to
the first component can send input to the second.

In attacking the system, the attacker selects a path in the graph that
begins with a designated \term{start vertex}~$s$.  Our results hold in
more general models (e.g., based on Horn clauses), but we defer discussing such
generalizations until Section~\ref{sec:generalizations}.  We think of the
attack as driving the system through the series of state transitions indicated
by the edges included in the path.  In the networking example in
Figure~\ref{fig:cloudy}, an attacker
might first compromise a front-end server and then leverage the
server's connectivity to the back-end database server to steal credit card
numbers from the database.

\paragraph{Incentives and Rewards.}
Attackers respond to incentives.  For example, attackers compromise machines
and form botnets because they make money from spam~\cite{spamalytics} or
rent the botnet to others~\cite{sabotageforhire}.  Other attackers steal
credit card numbers because credit card numbers have monetary
value~\cite{miscreants}.  
We model the attacker's incentives by attaching a non-negative \term{reward}
to each vertex.  These rewards are the utility the attacker derives from
driving the system into the state represented by the vertex.  For example,
compromising the database server might have a sizable reward because the
database server contains easily monetizable credit card numbers.  We assume
the start vertex has zero reward, forcing the attacker to undertake some action
before earning utility.  Whenever the attacker mounts an attack, the
attacker receives a \term{payoff} equal to the sum of the rewards of
the vertices visited in the attack path:
$\payoff(a) = \sum_{v \in a} \reward(a)$.
In the example from Figure~\ref{fig:cloudy}, if an attacker compromises both
a front-end server and the database server, the attacker receives both
rewards.

\paragraph{Attack Surface and Cost.}
The defender has a fixed \term{defense budget}~$B > 0$, which the defender
can divide among the edges in the graph according to a \term{defense
allocation}~$d$:
for all $e \in E$, $d(e) \geq 0$ and $\sum_{e \in E} d(e) \leq B$.

The defender's allocation of budget to various edges corresponds to the
decisions made by the Chief Information Security Officer~(CISO) about where
to allocate the enterprise's security resources.  For example, the CISO might
allocate organizational headcount to fuzzing enterprise web applications for
XSS vulnerabilities.  These kinds of investments are continuous in the sense
that the CISO can allocate $1/4$ of a full-time employee to worrying about
XSS.  We denote the set of feasible allocations of budget $B$ on edge set
$E$ by $\mathcal{D}_{B,E}$.  

By defending an edge, the defender makes it more difficult for the attacker
to use that edge in an attack.  Each unit of budget the defender allocates
to an edge raises the cost that the attacker must pay to use that edge in an
attack.  Each edge has an \term{attack surface}~\cite{attacksurface} $w$
that represents the difficulty in defending against that state transition.
For example, a server that runs both Apache and Sendmail has a larger attack
surface than one that runs only Apache because defending the first server is
more difficult than the second.  Formally, the attacker must pay the
following \term{cost} to traverse the edge:
$\cost(a, d) = \sum_{e \in a} d(e)/w(e)$.
Allocating defense budget to an edge does not ``reduce'' an edge's attack
surface.  For example, consider defending a hallway with bricks.  The wider
the hallway (the larger the attack surface), the more bricks (budget
allocation) required to build a wall of a certain height (the cost to the
attacker).

In this formulation, the function mapping the defender's budget allocation
to attacker cost is linear, preventing the defender from ever fully
defending an edge.  Our use of a linear function reflects a level of
abstraction more appropriate to a CISO who can never fully defend assets,
which we justify by observing that the rate of vulnerability discovery in a
particular piece of software is roughly constant~\cite{Rescorla}.
At a lower level of detail, we might replace this function with a step
function, indicating that the defender can ``patch'' a vulnerability by
allocating a threshold amount of budget.

\paragraph{Objective.}
To evaluate defense strategies, we measure the attacker's incentive for
attacking using the \term{return-on-attack}~(ROA)~\cite{ROA}, which we
define as follows:
\begin{gather*}
\ROA(a, d) = \frac{\payoff(a)}{\cost(a, d)}
\end{gather*}
We use this metric for evaluating defense strategy because we believe that
if the defender lowers the ROA sufficiently, the attacker will be
discouraged from attacking the system and will find other uses for his or
her capital or industry.  For example, the attacker might decide to attack
another system.  Analogous results hold if we quantify the attacker's
incentives in terms of profit (e.g., with $\profit(a, d) = \payoff(a) -
\cost(a, d)$), but we focus on ROA for simplicity.  

A purely rational attacker will mount attacks that maximize ROA.  However, a
real attacker might not maximize ROA.  For example, the attacker might not
have complete knowledge of the system or its defense.  We strengthen our
results by considering all attacks, not just those that maximize ROA.

\paragraph{Proactive Security.}
We evaluate our learning-based reactive approach by comparing it
against a \term{proactive} approach to risk management in which the defender
carefully examines the system and constructs a defense in order to fend off
future attacks.  We strengthen this benchmark by providing the proactive
defender complete knowledge about the system, but we require that the
defender commit to a fixed strategy.  To strengthen our results, we state
our main result in terms of \emph{all} such proactive defenders.  In
particular, this class of defenders includes the \term{rational proactive
defender} who employs a defense allocation that minimizes the maximum ROA
the attacker can extract from the system: $\argmin_d \max_a \ROA(a, d)$.

%\presec

\section{Case Studies} \label{sec:case-studies}

In this section, we describe instances of our model to build the reader's
intuition.  These examples illustrate that some familiar security concepts,
including perimeter defense and defense in depth, arise naturally as optimal
defenses in our model.  These defenses can be constructed either by rational
proactive attackers or converged to by a learning-based reactive defense.

\paragraph{Perimeter Defense.} \label{sec:perimeter-security}
Consider a system in which the attacker's reward is non-zero at exactly one
vertex, $t$.  For example, in a medical system, the attacker's reward for
obtaining electronic medical records might well dominate the value of other
attack targets such as employees' vacation calendars.  In such a system, a
rational attacker will select the minimum-cost path from the start vertex
$s$ to the valuable vertex $t$.  The optimal defense limits the attacker's
ROA by maximizing the cost of the minimum $s$-$t$ path.  The algorithm for
constructing this defense is straightforward~\cite{CMV06}:
\begin{enumerate}
\item Let $C$ be the minimum weight $s$-$t$ cut in $(V, E, w)$.

\item Select the following defense:
\begin{equation*}
d(e) = \begin{cases}
B w(e) / Z & \text{if $e \in C$} \\
0 & \text{otherwise}
\end{cases}\enspace,
\mbox{\ \ where\ } Z = \sum_{e \in C} w(e)\enspace.
\end{equation*}
\end{enumerate}
Notice that this algorithm constructs a \term{perimeter defense}: the
defender allocates the entire defense budget to a single cut in the graph.
Essentially, the defender spreads the defense budget over the attack surface
of the cut.  By choosing the minimum-weight cut, the defender is choosing to
defend the smallest attack surface that separates the start vertex from the
target vertex.  Real defenders use similar perimeter defenses, for example,
when they install a firewall at the boundary between their organization and
the Internet because the network's perimeter is much smaller than its
interior.

\paragraph{Defense in Depth.} \label{sec:defense-in-depth}
Many experts in security practice recommend that defenders employ defense in
depth.  %\cite{defenseindepth}
Defense in depth rises naturally in our model
as an optimal defense for some systems.  Consider, for example, the system
depicted in Figure~\ref{fig:simple-data-center-net}.  This attack graph is a
simplified version of the data center network depicted in
Figure~\ref{fig:cloudy}.  Although the attacker receives the largest reward
for compromising the back-end database server, the attacker also receives
some reward for compromising the front-end web server.  Moreover, the
front-end web server has a larger attack surface than the back-end database
server because the front-end server exposes a more complex interface (an
entire enterprise web application), whereas the database server exposes only
a simple SQL interface.  Allocating defense budget to the left-most edge
represents trying to protect sensitive database information with a complex
web application firewall instead of database access control lists (i.e.,
possible, but economically inefficient). 

\begin{figure}[t]
\centering
\includegraphics[width=0.5\columnwidth]{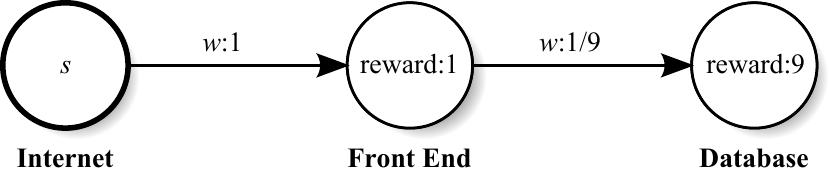}
\caption{Attack graph representing a simplified data center network.}
\label{fig:simple-data-center-net}
\end{figure}

The optimal defense against a rational attacker is to allocate half of the
defense budget to the left-most edge and half of the budget to the
right-most edge, limiting the attacker to a ROA of unity.
Shifting the entire budget to the right-most edge (i.e., defending
only the database) is disastrous because the attacker will simply attack the
front-end at zero cost, achieving an unbounded ROA.  Shifting the entire
budget to the left-most edge is also problematic because the
attacker will attack the database (achieving an ROA of $5$).

%\presec

\section{Reactive Security} \label{sec:reactive-security}

To analyze reactive security, we model the attacker and defender as playing
an iterative game, alternating moves.  First, the defender selects a
defense, and then the attacker selects an attack.  We present a
learning-based reactive defense strategy that is oblivious to vertex rewards
and to edges that have not yet been used in attacks.  We prove a theorem
bounding the competitive ratio between this reactive strategy and the best
proactive defense via a series of
reductions to results from the online learning theory literature.  Other
applications of this literature include managing stock
portfolios~\cite{portfolio}, playing zero-sum games~\cite{FS99}, and
boosting other machine learning heuristics~\cite{boost}.  Although we
provide a few technical extensions, our main contribution comes from
applying results from online learning to risk management.

%\subsection{Main Result}
%\label{sec:online-edges-defender}

%(However, we do assume the \emph{attacker} has complete knowledge of
%the system, including vulnerabilities, attack surfaces, and rewards).

\paragraph{Repeated Game.}
We formalize the repeated game between the defender and the attacker as
follows.  In each round $t$ from $1$ to $T$:
\begin{enumerate}
\item The defender chooses defense allocation $d_t(e)$ over the edges $e \in
E$.
\item The attacker chooses an attack path $a_t$ in $G$.
\item The path $a_t$ and attack surfaces $\left\{w(e) : e \in a_t\right\}$
are revealed to the defender.
\item The attacker pays $\cost(a_t, d_t)$ and gains $\payoff(a_t)$.
\end{enumerate}
In each round, we let the attacker choose the attack path after the defender
commits to the defense allocation because the defender's budget allocation is
not a secret (in the sense of a cryptographic key).  Following the ``no
security through obscurity'' principle, we make the conservative assumption
that the attacker can accurately determine the defender's budget allocation.

\paragraph{Defender Knowledge.}
Unlike proactive defenders, reactive defenders do not know all of the
vulnerabilities that exist in the system in advance.  (If defenders had
complete knowledge of vulnerabilities, conferences such as Black Hat
Briefings %\cite{BlackHat}
would serve little purpose.) Instead, we reveal an
edge (and its attack surface) to the defender after the attacker uses the
edge in an attack.  For
example, the defender might monitor the system and learn how the attacker
attacked the system by doing a post-mortem analysis of intrusion logs.
Formally, we define a \term{reactive defense
strategy} to be a function from attack sequences $\{a_i\}$ and the subsystem
induced by the edges contained in $\bigcup_i a_i$ to defense allocations such
that $d(e) = 0$ if edge $e \not \in \bigcup_i a_i$.  Notice that this
requires the defender's strategy to be oblivious to the system beyond the
edges used by the attacker.

\begin{algorithm}
\begin{itemize}
\item Initialize $E_0 = \emptyset$
\item For each round $t \in \{2,...,T\}$
\begin{itemize}
\item Let $E_{t-1} = E_{t-2}\cup E(a_{t-1})$
\item For each $e\in E_{t-1}$, let
\begin{align*}
S_{t-1}(e) &= \begin{cases}
S_{t-2}(e)+M(e,a_{t-1}) & \text{if $e\in E_{t-2}$} \\
M(e,a_{t-1})            & \text{otherwise.} \\
\end{cases} \\
\tilde{P}_t(e) &= \beta_{t-1}^{S_{t-1}(e)} \\
P_t(e) &= \frac{\tilde{P}_t(e)}{\sum_{e'\in E_t}\tilde{P}_t(e')}\enspace,
\end{align*} \\
where $M(e,a) = -\indic{e\in a}/w(e)$ is a matrix with $|E|$ rows and a column for each attack.
\end{itemize}
\end{itemize}
\caption{A reactive defense strategy for hidden edges.}
\label{alg:mw-online-edges}
\end{algorithm}

\vfill\eject

\paragraph{Algorithm.}
Algorithm~\ref{alg:mw-online-edges} is a reactive defense strategy based on
the multiplicative update learning algorithm~\cite{CFHHSW97,FS99}.  
The algorithm reinforces edges on the attack path multiplicatively, taking
the attack surface into account by allocating more budget to
easier-to-defend edges.  When new edges are revealed, the algorithm
re-allocates budget uniformly from the already-revealed edges to the newly
revealed edges.  We state the algorithm in terms of a normalized defense
allocation $P_t(e) = d_t(e)/B$.  Notice that this algorithm is oblivious to
unattacked edges and the attacker's reward for visiting each vertex.  An
appropriate setting for the algorithm parameters $\beta_t \in [0,1)$ will be
described below.

The algorithm begins without any knowledge of the graph whatsoever, and so
allocates no defense budget to the system.  Upon the \nth{t}{th} attack on
the system, the algorithm updates $E_t$ to be the set of edges revealed up
to this point, and updates $S_t(e)$ to be a weight count of the number of
times $e$ has been used in an attack thus far.  For each edge that has ever
been revealed, the defense allocation $P_{t+1}(e)$ is chosen to be
$\beta_t^{S_t(e)}$ normalized to sum to unity over all edges $e \in E_t$.
In this way, any edge attacked in round $t$ will have its defense allocation
reinforced.

The parameter $\beta$ controls how aggressively the defender reallocates
defense budget to recently attacked edges.  If $\beta$ is infinitesimal, the
defender will move the entire defense budget to the edge on the most recent
attack path with the smallest attack surface.  If $\beta$ is enormous, the
defender will not be very agile and, instead, leave the defense budget in
the initial allocation.  For an appropriate value of $\beta$, the algorithm
will converge to the optimal defense strategy.  For instance, the min cut in
the example from Section~\ref{sec:perimeter-security}.

\paragraph{Theorems.}
To compare this reactive defense strategy to all proactive defense
strategies, we use the notion of \term{regret} from online learning theory.
The following is an additive regret bound relating the attacker's profit
under reactive and proactive defense strategies.

\begin{thm} \label{thm:profit-regret-online}
The average attacker profit against Algorithm~\ref{alg:mw-online-edges}
converges to the average attacker profit against the best proactive defense.
Formally, if defense allocations $\{d_t\}_{t=1}^T$ are output by
Algorithm~\ref{alg:mw-online-edges} with parameter sequence
$\beta_s = \left(1+\sqrt{2\log|E_s|/(s+1)}\right)^{-1}$
on any system $(V,E,w,\reward,s)$ revealed online and any attack sequence
$\{a_t\}_{t=1}^T$, then 
\begin{eqnarray*}
\frac{1}{T}\sum_{t=1}^T\profit(a_t,d_t) - \frac{1}{T}\sum_{t=1}^T\profit(a_t,d^\star) &\leq& B\sqrt{\frac{\log|E|}{2T}}+\frac{B(\log|E|+\overline{w^{-1}})}{T}\enspace,
\end{eqnarray*}
for all proactive defense strategies $d^\star \in \mathcal{D}_{B,E}$ where
$\overline{w^{-1}} = |E|^{-1}\sum_{e \in E} w(e)^{-1}$, the mean of the
surface reciprocals.
\end{thm}

\begin{rem}
We can interpret Theorem~\ref{thm:profit-regret-online} as establishing
sufficient conditions under which a
reactive defense strategy is within an additive constant of the best
proactive defense strategy.  Instead of carefully analyzing the system to
construct the best proactive defense, the defender need only react to
attacks in a principled manner to achieve almost the same quality of
defense in terms of attacker profit.
\end{rem}
Reactive defense strategies can also be competitive with proactive defense
strategies when we consider an attacker motivated by return on
attack~(ROA).  The ROA formulation is appealing because (unlike with
profit) the objective function does not require measuring attacker cost and
defender budget in the same units.  The next result considers the
competitive ratio between the ROA for a reactive defense strategy and the
ROA for the best proactive defense strategy.
\begin{thm} \label{thm:ROA-competitive-online}
The ROA against Algorithm~\ref{alg:mw-online-edges} converges to the ROA
against best proactive defense.  Formally, consider the cumulative ROA:
\begin{gather*}
\ROA\left(\{a_t\}_{t=1}^T,\{d_t\}_{t=1}^T\right) =
\frac{\sum_{t=1}^T\payoff(a_t)}{\sum_{t=1}^T\cost(a_t, d_t)}
\end{gather*}
(We abuse notation slightly and use singleton arguments to represent the
corresponding constant sequence.) If defense allocations $\{d_t\}_{t=1}^T$
are output by Algorithm~\ref{alg:mw-online-edges} with parameters
$\beta_s=\left(1+\sqrt{2\log|E_s|/(s+1)}\right)^{-1}$
on any system $(V,E,w,\reward,s)$ revealed online, such that $|E|>1$, and
any attack sequence $\{a_t\}_{t=1}^T$, then for all $\alpha > 0$ and
proactive defense strategies $d^\star\in\mathcal{D}_{B,E}$
\begin{gather*}
\frac{\ROA\left(\{a_t\}_{t=1}^T,\{d_t\}_{t=1}^T\right)}
{\ROA\left(\{a_t\}_{t=1}^T,d^\star\right)}
\leq 1+\alpha\enspace,
\end{gather*}
provided $T$ is sufficiently large.\footnote{To wit: $T \geq
\left(\frac{13}{\sqrt{2}}\left(1+\alpha^{-1}\right)\left(\sum_{e\in
\inc(s)}w(e)\right)\right)^2\log|E|$.}
\end{thm}

\begin{rem}
Notice that the reactive defender can use \emph{the same algorithm}
regardless of whether the attacker is motivated by profit or by ROA.  As
discussed in Section~\ref{sec:profit-v-roa} the optimal proactive defense is
not similarly robust.
\end{rem}
We present proofs of these theorems in Appendix~\ref{apx:proofs}.  We first
prove the theorems in the simpler setting where the defender knows the
entire graph.  Second, we remove the hypothesis that the defender knows the
edges in advance.

\paragraph{Lower Bounds.}
In Appendix~\ref{apx:proofs}, we use a two-vertex, two-edge graph to
establish a lower bound on the competitive ratio of the ROA for all reactive
strategies.  The lower bound shows that the analysis of
Algorithm~\ref{alg:mw-online-edges} is tight and that
Algorithm~\ref{alg:mw-online-edges} is optimal given the information
available to the algorithm.  The proof gives an example where the best
proactive defense (slightly) out-performs every reactive strategy,
suggesting the benchmark is not unreasonably weak.

%\presec

\section{Advantages of Reactivity} \label{sec:advantages}

In this section, we examine some
situations in which a reactive defender out-performs a proactive defender.
Proactive defenses hinge on the defender's model of the attacker's
incentives.  If the defender's model is inaccurate, the defender will
construct a proactive defense that is far from optimal.  By contrast, a
reactive defender need not reason about the attacker's incentives directly.
Instead, the reactive defender learns these incentives by observing
the attacker in action.

\paragraph{Learning Rewards.}
One way to model inaccuracies in the defender's estimates of the attacker's
incentives is to hide the attacker's rewards from the defender.  Without
knowledge of the payoffs, a proactive defender has difficulty limiting the
attacker's ROA.  Consider, for example, the star system whose edges have
equal attack surfaces, as depicted in Figure~\ref{fig:star}.  Without
knowledge of the attacker's rewards, a proactive defender has little choice
but to allocate the defense budget equally to each edge (because the edges
are indistinguishable).  However, if the attacker's reward is concentrated
at a single vertex, the competitive ratio for attacker's ROA (compared to
the rational proactive defense) is the number of leaf vertices.  (We can, of
course, make the ratio worse by adding more vertices.)  By contrast, the
reactive algorithm we analyze in Section~\ref{sec:reactive-security} is
competitive with the rational proactive defense because the reactive
algorithm effectively learns the rewards by observing which attacks the
attacker chooses.

\begin{figure}[t]
\begin{minipage}[t]{0.4\textwidth}
\vspace{0pt}
\includegraphics[width=\linewidth]{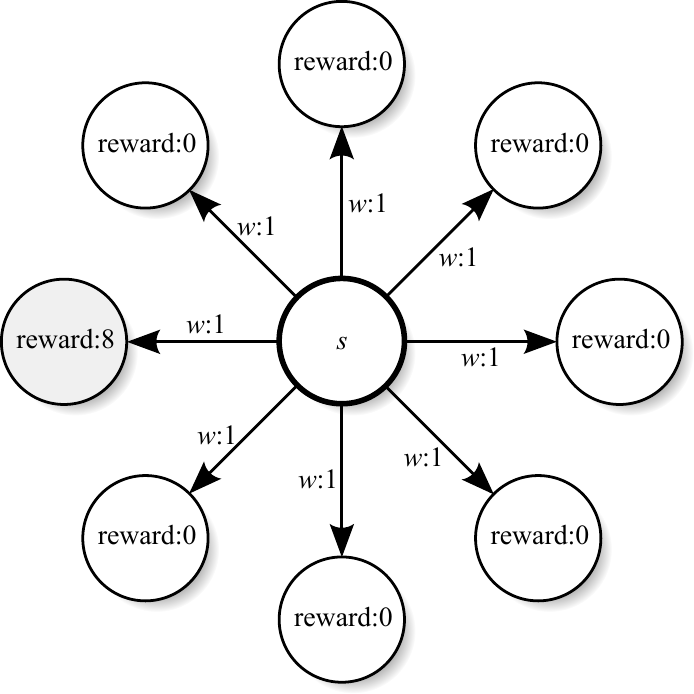}
\caption{Star-shaped attack graph with rewards concentrated in an unknown
vertex.}
\label{fig:star}
\end{minipage}
\hfill
\begin{minipage}[t]{0.5\textwidth}
\vspace{1.94cm}
\includegraphics[width=\linewidth]{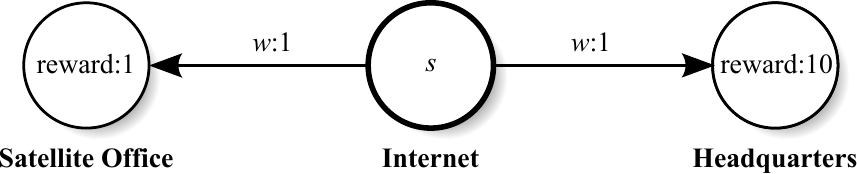}
\vspace{1.3cm}
\caption{An attack graph that separates the minimax strategies optimizing
ROA and attacker profit.}
\label{fig:left-right}
\end{minipage}
\end{figure}    

\paragraph{Robustness to Objective.} \label{sec:profit-v-roa}
Another way to model inaccuracies in the defender's estimates of the
attacker's incentives is to assume the defender mistakes which of profit and
ROA actually matter to the attacker.  The defense constructed by a rational
proactive defender depends crucially on whether the attacker's actual
incentives are based on profit or based on ROA, whereas the reactive
algorithm we analyze in Section~\ref{sec:reactive-security} is robust to this
variation.  In particular, consider the system depicted in
Figure~\ref{fig:left-right}, and assume the defender has a budget of $9$.
If the defender believes the attacker is motivated by profit, the rational
proactive defense is to allocate the entire defense budget to the right-most
edge (making the profit $1$ on both edges).  However,
this defense is disastrous when viewed in terms of ROA because the ROA for
the left edge is infinite (as opposed to near unity when the proactive
defender optimizes for ROA).

\paragraph{Catachresis.}
The defense constructed by the rational proactive defender is optimized for
a rational attacker.  If the attacker is not perfectly rational, there is
room for out-performing the rational proactive defense.  There are a number
of situations in which the attacker might not mount ``optimal'' attacks:
\begin{itemize}
\item The attacker might not have complete knowledge of the attack graph.
Consider, for example, a software vendor who discovers five equally severe
vulnerabilities in one of their products via fuzzing.  According to
proactive security, the defender ought to dedicate equal
resources to repairing these five vulnerabilities.  However, a reactive
defender might dedicate more resources to fixing a vulnerability actually
exploited by attackers in the wild.  We can model these situations by
making the attacker oblivious to some edges.

\item The attacker might not have complete knowledge of the defense
allocation.  For example, an attacker attempting to invade a corporate
network might target computers in human resources without realizing that
attacking the customer relationship management database in sales has a
higher return-on-attack because the database is lightly defended.
\end{itemize}
By observing attacks, the reactive strategy learns a defense tuned for
the \emph{actual} attacker, causing the attacker to receive a lower ROA.

% In these situations, a defender could hope for a defense that out-performs
% the rational proactive defense.  Formally, consider the example shown in
% Figure~\ref{fig:left-right}.  If the defender's budget is $11$, then the
% rational proactive defense (for ROA) is to allocate $10$ to defending the
% right-most edge and $1$ to defending the left-most edge.  Suppose, however,
% that the attacker always attacks the left edge.  The rational proactive
% defense limits the ROA to $1$, but a reactive defender achieves a
% ROA limit approaching $1/11$ exponentially fast because the reactive
% defender reacts to the attacker's apparent ignorance of the headquarters
% prize by shifting the defense budget to the left.

%\presec

\section{Generalizations} \label{sec:generalizations}

% \todo{benr: flesh this out.}
% Of course any prior
% knowledge of a subsystem $F$ of edges can easily be incorporated into the
% algorithm by initializing $E_0=F$ and $P_1(e)=|F|^{-1}$ for each $e\in F$,
% with the corresponding bounds on profit and ROA improving accordingly.
% \todo{by how much? Consider moving as a subsection in the generalizations
% section}

\paragraph{Horn Clauses.}
Thus far, we have presented our results using a graph-based system model.
Our results extend, however, to a more general system model based on Horn
clauses.  Datalog programs, which are based on Horn clauses, have been used
in previous work to represent vulnerability-level attack
graphs~\cite{OBM06}.  A Horn clause is a statement in propositional logic of
the form $p_1 \land p_2 \land \cdots \land p_n \to q.$ The propositions
$p_1, p_2, \ldots, p_n$ are called the \term{antecedents}, and $q$ is called
the \term{consequent}.  The set of antecedents might be empty, in which case
the clause simply asserts the consequent.  Notice that Horn clauses are
negation-free.  In some sense, a Horn clause represents an edge in a
hypergraph where multiple pre-conditions are required before taking a
certain state transition.

In the Horn model, a system consists of a set of Horn clauses, an attack
surface for each clause, and a reward for each proposition.  The defender
allocates defense budget among the Horn clauses.  To mount an attack, the
attacker selects a \term{valid proof}: an ordered list of rules such that
each antecedent appears as a consequent of a rule earlier in the list.  For
a given proof $\Pi$,
\begin{align*}
\cost(\Pi, d) &= \sum_{c \in \Pi} d(c)/w(e) &
\payoff(\Pi) &= \sum_{p \in \brak{\Pi}} \reward(p)\enspace,
\end{align*}
where $\brak{\Pi}$ is the set of propositions proved by $\Pi$ (i.e., those
propositions that appear as consequents in $\Pi$).  Profit and ROA are
computed as before.

Our results generalize to this model directly.  Essentially, we need only
replace each instance of the word ``edge'' with ``Horn clause'' and ``path''
with ``valid proof.''  For example, the rows of the matrix $M$ used
throughout the proof become the Horn clauses, and the columns become the
valid proofs (which are numerous, but no matter).  The entries of the matrix
become $M(c, \Pi) = 1/w(c)$, analogous to the graph case.  The one
non-obvious substitution is $\inc(s)$, which becomes the set of clauses that
lack antecedents.

\paragraph{Multiple Attackers.}
We have focused on a security game between a \emph{single} attacker and a
defender.  In practice, a security system might be attacked by
several uncoordinated attackers, each with different information
and different objectives.  Fortunately, we can show that a model with
multiple attackers is mathematically equivalent to a model with a single
attacker with a randomized strategy: Use the set of attacks, one per
attacker, to define a distribution over edges where the probability of an
edge is linearly proportional to the number of attacks which use the edge.
This precludes the interpretation of an attack as an $s$-rooted path, but
our proofs do not rely upon this interpretation and our results hold in
such a model with appropriate modifications.

\paragraph{Adaptive Proactive Defenders.}
A simple application of an online learning result~\cite{HW98}, omitted due
to space constraints, modifies our regret bounds for a proactive defender
who re-allocates budget a fixed number of times.  In this model, our results
remain qualitatively the same.

%\presec

\section{Related Work} \label{sec:related-work}

Anderson~\cite{Anderson2001} and Varian~\cite{Varian2000} informally discuss
(via anecdotes) how the design of information security must take incentives
into account.  August and Tunca~\cite{August2006} compare various ways to
incentivize users to patch their systems in a setting where the users are
more susceptible to attacks if their neighbors do not patch.

Gordon and Loeb~\cite{Gordon2002} and Hausken~\cite{Hausken} analyze the
costs and benefits of security in an economic model (with non-strategic
attackers) where the probability of a successful exploit is a function of
the defense investment.  They use this model to compute the optimal level of
investment.  Varian~\cite{Varian2004} studies various (single-shot) security
games and identifies how much agents invest in security at
equilibrium.  Grossklags~\cite{Grossklags2008} extends this model by
letting agents self-insure. 

Miura et~al.~\cite{Miura2008} study externalities that appear due to
users having the same password across various websites and discuss
pareto-improving security investments.  Miura and Bambos~\cite{Miura2007}
rank vulnerabilities according to a random-attacker model.
Skybox %\cite{skybox}
and RedSeal %\cite{redseal}
offer practical systems that
help enterprises prioritize vulnerabilities based on a
random-attacker model.  Kumar et~al.~\cite{Kumar} investigate optimal
security architectures for a multi-division enterprise, taking into account
losses due to lack of availability and confidentiality.  None of the
above papers explicitly model a truly adversarial attacker.

Fultz~\cite{Fultz} generalizes~\cite{Grossklags2008} by modeling attackers
explicitly.  Cavusoglu et~al.~\cite{Cavusoglu2008} highlight the importance
of using a game-theoretic model over a decision theoretic model due to the
presence of adversarial attackers.  However, these models look at idealized
settings that are not generically applicable.  Lye and Wing~\cite{Wing}
study the Nash equilibrium of a single-shot game between an attacker and a
defender that models a particular enterprise security scenario.  Arguably
this model is most similar to ours in terms of abstraction level.  However,
calculating the Nash equilibrium requires detailed knowledge of the
adversary's incentives, which as discussed in the introduction, might not be
readily available to the defender.  Moreover, their game contains multiple
equilibria, weakening their prescriptions.

%\presec

\section{Conclusions} \label{sec:conclusions}

Many security experts equate reactive security with myopic bug-chasing and
ignore principled reactive strategies when they recommend adopting a
proactive approach to risk management.  In this paper, we establish
sufficient conditions for a learning-based reactive strategy to be
competitive with the best fixed proactive
defense.  Additionally, we show that reactive defenders can out-perform
proactive defenders when the proactive defender defends against attacks that
never actually occur.
Although our model is an abstraction of the complex interplay between
attackers and defenders, our results support the following practical advice
for CISOs making security investments:
\begin{itemize}
\item Employ monitoring tools that let you detect and analyze attacks
against your enterprise.  These tools help focus your efforts on thwarting
real attacks.

\item Make your security organization more agile.  For example, build a
rigorous testing lab that lets you roll out security patches quickly once
you detect that attackers are exploiting these vulnerabilities.

\item When determining how to expend your security budget, avoid overreacting
to the most recent attack.  Instead, consider all previous attacks, but
discount the importance of past attacks exponentially.
\end{itemize}
In some situations, proactive security can out-perform reactive security.
For example, reactive approaches are ill-suited for defending against
catastrophic attacks because there is no ``next round'' in which the
defender can use information learned from the attack.  We hope our results
will lead to a productive discussion of the limitations of our model and the
validity of our conclusions.

Instead of assuming that proactive security is always superior to reactive
security, we invite the reader to consider when a reactive approach might be
appropriate.  For the parts of an enterprise where the defender's budget is
liquid and there are no catastrophic losses, a carefully constructed
reactive strategy can be as effective as the best proactive defense in the
worst case and significantly better in the best case.

\paragraph{Acknowledgments}
We would like to thank Elie Bursztein, Eu-Jin Goh, and Matt Finifter for their
thoughtful comments and helpful feedback. We gratefully acknowledge the support
of the NSF through the TRUST Science and Technology Center and grants
DMS-0707060, CCF-0424422, 0311808, 0448452, and 0627511, and the support of the
AFOSR through the MURI Program, and the support of the Siebel Scholars
Foundation.

%\presec

\bibliography{react}

\newpage
\appendix

%\presec

\section{Proofs}
\label{apx:proofs}

We now describe a series of reductions that establish the main results.
First, we prove Theorem~\ref{thm:profit-regret-online} in the simpler
setting where the defender knows the entire graph.  Second, we remove the
hypothesis that the defender knows the edges is advance.  Finally, we extend
our results to ROA.

\paragraph{Profit (Known Edges).}
% Our particular choice of game matrix extends the matrix used
% in~\cite{CMV06}. The $M(e, a)$ entry in the matrix is $-1/w(e)$ if $e \in
% a$ and zero otherwise.  In this interpretation, the defender's defense
% allocation in each round (suitably deflated by the defense budget $B$) is
% a \todo{probability distribution} over the edges.  With a slight abuse of
% notation, we write $M(P, a)$ for the expected value of the game when a
% random edge is drawn according to \todo{distribution} $P$.  This quantity
% is the negative of the attacker's cost for $a$ under defense allocation
% $P$ (scaled by $B$).
% 
Suppose that the reactive defender is granted full knowledge of the system
$(V, E, w, \reward, s)$ from the outset.  Specifically, the graph, attack
surfaces, and rewards are all revealed to the defender prior to the first
round.  Algorithm~\ref{alg:mw-basic} is a reactive defense strategy that
makes use of this additional knowledge.

\begin{algorithm}
\begin{itemize}
\item For each $e \in E$, initialize $P_1(e) = 1/|E|$. 
\item For each round $t \in \{2, \ldots, T\}$ and $e \in E$, let
\begin{align*}
P_t(e) &= P_{t-1}(e) \cdot \beta^{M(e,a_{t-1})}/{Z_{t}} \\
\mbox{where\hspace{3em}}Z_{t} &= \sum_{e'\in E}P_{t-1}(e)\beta^{M(e',a_{t-1})}
\end{align*}
\end{itemize}
\caption{Reactive defense strategy for known edges using the multiplicative
update algorithm.}
\label{alg:mw-basic}
\end{algorithm}

\begin{lem} \label{lemma:profit-regret-basic}
If defense allocations $\{d_t\}_{t=1}^T$ are output by
Algorithm~\ref{alg:mw-basic} with parameter
$\beta = \left(1+\sqrt{\frac{2\log|E|}{T}}\right)^{-1}$
on any system $(V,E,w,\reward,s)$ and attack sequence $\{a_t\}_{t=1}^T$,
then 
\begin{eqnarray*}
\frac{1}{T}\sum_{t=1}^T\profit(a_t,d_t) - \frac{1}{T}\sum_{t=1}^T\profit(a_t,d^\star) &\leq& B\sqrt{\frac{\log|E|}{2T}}+\frac{B\log|E|}{T}\enspace,
\end{eqnarray*}
for all proactive defense strategies $d^\star \in \mathcal{D}_{B,E}$.
\end{lem}

The lemma's proof is a reduction to the following regret bound from online
learning~\cite[Corollary 4]{FS99}.
\begin{thm} \label{thm:additive-regret} If the multiplicative update algorithm
(Algorithm~\ref{alg:mw-basic}) is run with any game matrix $M$ with elements in
$[0,1]$, and parameter $\beta = \left(1+\sqrt{2\log|E|/T}\right)^{-1}$,
then
\begin{eqnarray*}
\frac{1}{T}\sum_{t=1}^T M(P_t,a_t)-\min_{P^\star\geq0:\sum_{e\in E}P^\star(e)=1}\left\{\frac{1}{T}\sum_{t=1}^T M(P^\star,a_t)\right\} 
&\leq& \sqrt{\frac{\log|E|}{2T}} + \frac{\log|E|}{T}\enspace.
\end{eqnarray*}
\end{thm}
\begin{proof}[of Lemma~\ref{lemma:profit-regret-basic}]
Due to the normalization by $Z_t$, the sequence of defense allocations
$\{P_t\}_{t=1}^T$ output by Algorithm~\ref{alg:mw-basic} is invariant to
adding a constant to all elements of matrix $M$.  Let $M'$ be the matrix
obtained by adding constant $C$ to all entries of arbitrary game matrix $M$,
and let sequences $\{P_t\}_{t=1}^T$ and $\{P_t'\}_{t=1}^T$ be obtained by
running multiplicative update with matrix $M$ and $M'$ respectively.  Then,
for all $e \in E$ and $t \in [T-1]$,
\begin{align*}
P_{t+1}'(e) &= \frac{P_1(e)\beta^{\sum_{i=1}^t M'(e,a_i)}}
        {\sum_{e'\in E}P_1(e') \beta^{\sum_{i=1}^t M'(e',a_i)}} &=& \frac{P_1(e)\beta^{\left(\sum_{i=1}^t M(e,a_i)\right)+tC}}{\sum_{e'\in E}P_1(e')\beta^{\left(\sum_{i=1}^t M(e',a_i)\right)+tC}} \\
&= \frac{P_1(e)\beta^{\sum_{i=1}^t M(e,a_i)}}{\sum_{e'\in E}P_1(e')\beta^{\sum_{i=1}^t M(e',a_i)}} &=& P_{t+1}(e)\enspace.
\end{align*}
In particular Algorithm~\ref{alg:mw-basic} produces the same defense
allocation sequence as if the game matrix elements are increased by one to
\begin{gather*}
M'(e,a) = \begin{cases}
1-1/w(e) & \text{if $e \in a$} \\
1        & \text{otherwise.}
\end{cases}
\end{gather*}
Because this new matrix has entries in $[0,1]$ we can apply
Theorem~\ref{thm:additive-regret} to prove for the original matrix $M$ that
\begin{eqnarray}
\frac{1}{T}\sum_{t=1}^T M(P_t,a_t) -
\min_{P^\star\in\mathcal{D}_{1,E}}
\left\{\frac{1}{T}\sum_{t=1}^T M(P^\star,a_t)\right\} 
&\leq& \sqrt{\frac{\log|E|}{2T}}+\frac{\log|E|}{T}.
\label{eq:thm-add-regret}
\end{eqnarray}
Now, by definition of the original game matrix,
\begin{align*}
M(P_t,a_t) &= \sum_{e \in E}-(P_t(e)/w(e))\cdot\indic{e \in a_t} &=& -\sum_{e \in a_t}P_t(e)/w(e)\\
 &= -B^{-1}\sum_{e\in a_t}d_t(e)/w(e) &=& -B^{-1}\cost(a_t,d_t)\enspace.
\end{align*}
Thus Inequality~\eqref{eq:thm-add-regret} is equivalent to
\begin{eqnarray*}
&& -\frac{1}{T}\sum_{t=1}^T B^{-1} \cost(a_t,d_t) - \min_{d^\star\in\mathcal{D}_{1,E}} \left\{ -\frac{1}{T}\sum_{t=1}^T B^{-1}\cost(a_t,d^\star)\right\} \\
&\leq& \sqrt{\frac{\log|E|}{2T}}+\frac{\log|E|}{T}
\end{eqnarray*}
Simple algebraic manipulation yields
\begin{eqnarray*}
&& \frac{1}{T}\sum_{t=1}^T\profit(a_t,d_t)-\min_{d^\star\in\mathcal{D}_{B,E}}
\left\{ \frac{1}{T}\sum_{t=1}^T\profit(a_t,d^\star)\right\} \\
&=& \frac{1}{T}\sum_{t=1}^T\left(\payoff(a_t)-\cost(a_t,d_t)\right)- 
\min_{d^\star\in\mathcal{D}_{B,E}}
\left\{ \frac{1}{T}\sum_{t=1}^T\left(\payoff(a_t)-\cost(a_t,d^\star)\right)\right\} \\
&=& \frac{1}{T}\sum_{t=1}^T\left(-\cost(a_t,d_t)\right)- 
\min_{d^\star\in\mathcal{D}_{B,E}}
\left\{ \frac{1}{T}\sum_{t=1}^T\left(-\cost(a_t,d^\star)\right)\right\} \\
&\leq& B\sqrt{\frac{\log|E|}{2T}}+B\frac{\log|E|}{T}\enspace.
\end{eqnarray*}
\end{proof}

% Consider the numerator in the dominant term of
% Lemma~\ref{lemma:profit-regret-basic}.
% If we consider a real-world enterprise as allocating security budget
% $B$ proportional to the number of assets $O(|V|)$, then the numerator
% $B\sqrt{\log|E|}$ is $O\left(|V|\sqrt{\log|V|}\right)$.

\paragraph{Profit (Hidden Edges).}
The standard algorithms in online learning assumes that the rows of the
matrix are known in advance.  Here, the edges are not known in advance and
we must relax this assumption using a simulation argument, which is perhaps
the least obvious part of the reduction.
The defense allocation chosen by Algorithm~\ref{alg:mw-online-edges} at time
$t$ is precisely the same as the defense allocation that would have been
chosen by Algorithm~\ref{alg:mw-basic} had the defender run
Algorithm~\ref{alg:mw-basic} on the currently visible subgraph.  The
following lemma formalizes this equivalence.  Note that
Algorithm~\ref{alg:mw-online-edges}'s parameter is reactive: it corresponds
to the Algorithm~\ref{alg:mw-basic}'s parameter, but for the subgraph
induced by the edges revealed so far.  That is, $\beta_t$ depends only on
edges visible to the defender in round $t$, letting the defender actually
run the algorithm.

\begin{lem}
Consider arbitrary round $t \in [T]$. If
Algorithms~\ref{alg:mw-online-edges} and~\ref{alg:mw-basic} are run with
parameters
$\beta_s =\left(1+\sqrt{2\log|E_s|/(s+1)}\right)^{-1}$
for $s\in[t]$ and parameter
$\beta = \left(1+\sqrt{2\log|E_t|/(t+1)}\right)^{-1}$
 respectively, with the latter run on the subgraph induced by
$E_t$, then the defense allocations $P_{t+1}(e)$ output by the
algorithms are identical for all $e \in E_t$.
\end{lem}
\begin{proof}
If $e \in E_t$ then $\tilde{P}_{t+1}(e) = \beta^{\sum_{i=1}^t M(e,a_i)}$
because $\beta_t=\beta$, and the round $t+1$ defense allocation of
Algorithm~\ref{alg:mw-online-edges} $P_{t+1}$ is simply $\tilde{P}_{t+1}$
normalized to sum to unity over edge set $E_t$, which is exactly the defense
allocation output by Algorithm \ref{alg:mw-basic}.
\end{proof}
Armed with this correspondence, we show that
Algorithm~\ref{alg:mw-online-edges} is almost as effective as
Algorithm~\ref{alg:mw-basic}.  In other words, hiding unattacked edges from
the defender does not cause much harm to the reactive defender's ability to
disincentivize the attacker.
\begin{lem}
\label{lem:relating-alg-1-2}
If defense allocations $\{d_{1,t}\}_{t=1}^T$ and $\{d_{2,t}\}_{t=1}^T$ are
output by Algorithms~\ref{alg:mw-online-edges} and~\ref{alg:mw-basic}
with parameters
$\beta_t = \left(1+\sqrt{2\log|E_t|/(t+1)}\right)^{-1}$
for $t\in[T-1]$ and
$\beta = \left(1+\sqrt{2\log|E|/(T)}\right)^{-1}$,
respectively, on a system $(V,E,w,\reward,s)$ and attack sequence
$\{a_t\}_{t=1}^T$, then
\begin{eqnarray*}
\frac{1}{T}\sum_{t=1}^T\profit(a_{t},d_{1,t})
-\frac{1}{T}\sum_{t=1}^T\profit(a_t,d_{2,t})
& \leq & \frac{B}{T}\overline{w^{-1}}\enspace.
\end{eqnarray*}
\end{lem}
\begin{proof}
Consider attack $a_t$ from a round $t \in [T]$ and consider an edge $e \in
a_t$.  If $e \in a_s$ for some $s < t$, then the defense budget allocated to
$e$ at time $t$ by Algorithm~\ref{alg:mw-basic} cannot be greater than the
budget allocated by Algorithm~\ref{alg:mw-online-edges}.
Thus, the instantaneous cost paid by the attacker on $e$ when
Algorithm~\ref{alg:mw-online-edges} defends is at least the cost paid when
Algorithm~\ref{alg:mw-basic} defends: $d_{1,t}(e)/w(e) \geq d_{2,t}(e)/w(e)$.
If $e \notin \bigcup_{s=1}^{t-1} a_s$ then for all $s \in [t]$, $d_{1,s}(e)
= 0$, by definition.  The sequence $\{d_{2,s}(e)\}_{s=1}^{t-1}$ is decreasing
and positive.  Thus $\max_{s<t}d_{2,s}(e)-d_{1,s}(e)$ is optimized
at $s = 1$ and is equal to $B/|E|$.  Finally because each edge $e \in E$
is first revealed exactly once this leads to
\begin{eqnarray*}
\sum_{t=1}^T\cost(a_t,d_{2,t})-\sum_{t=1}^T\cost(a_t,d_{1,t})
= \sum_{t=1}^T\sum_{e\in a_t}\frac{d_{2,t}(e)-d_{1,t}(e)}{w(e)}
\leq \sum_{e \in E}\frac{B}{|E|w(e)}\enspace.
\end{eqnarray*}
Combined with the fact that the attacker receives the same payout whether
Algorithm~\ref{alg:mw-basic} or Algorithm~\ref{alg:mw-online-edges} defends
completes the result.
\end{proof}

\begin{proof}[of Theorem~\ref{thm:profit-regret-online}]
The result follow immediately from Lemma~\ref{lemma:profit-regret-basic} and 
Lemma~\ref{lem:relating-alg-1-2}.
\end{proof}

Finally, notice that Algorithm~\ref{alg:mw-online-edges} enjoys the same time
and space complexities as Algorithm~\ref{alg:mw-basic}, up to constants.

\paragraph{ROA (Hidden Edges).}
We now translate our bounds on profit into bounds on ROA by observing that
ratio of two quantities is small if the quantities are large and their
difference is small.  We consider the competitive ratio between an reactive
defense strategy and the best proactive defense strategy after the following
technical lemma, which asserts that the quantities are large.
\begin{lem} \label{lem:max-gt-value}
For all attack sequences $\{a_t\}_{t=1}^T$,
$\max_{d^\star\in\mathcal{D}_{B,E}}\sum_{t=1}^T\cost(a_t,d^\star) \geq VT$
where game value $V$ is
$\max_{d\in\mathcal{D}_{B,E}}\min_{a}\cost(a,d)
= \frac{B}{\sum_{e\in \inc(s)}w(e)} > 0$,
where $\inc(v) \subseteq E$ denotes the edges incident to vertex $v$.
\end{lem}
\begin{proof}
Let $d^\star = \argmax_{d\in\mathcal{D}_{B,E}}\min_a\cost(a,d)$ witness the
game's value $V$, then
$\max_{d\in\mathcal{D}_{B,E}}\sum_{t=1}^T\cost(a_t,d)\enspace
\geq\enspace \sum_{t=1}^T\cost(a_t,d^\star) \enspace
\geq\enspace TV$.
Consider the defensive allocation for each $e \in E$.  If $e \in
\inc(s)$, let $\tilde{d}(e)= B w(e)/\sum_{e \in \inc(s)} w(e) > 0$, and
otherwise $\tilde{d}(e)=0$.  This allocation is feasible because
\begin{gather*}
\sum_{e\in E}\tilde{d}(e)
= \frac{B\sum_{e \in \inc(s)}w(e)}{\sum_{e \in \inc(s)}w(e)}
= B\enspace.
\end{gather*}
By definition $\tilde{d}(e)/w(e)=B/\sum_{e\in \inc(s)}w(e)$
for each edge $e$ incident to $s$.  Therefore, $\cost(a,\tilde{d})\geq
B/\sum_{e\in \inc(s)}w(e)$ for any non-trivial attack $a$, which
necessarily includes at least one $s$-incident edge.  Finally,
$V \geq \min_{a}\cost(a,\tilde{d})$ proves
\begin{gather}
V \geq \frac{B}{\sum_{e\in \inc(s)}w(e)}\enspace. \label{eq:V-lower}
\end{gather}
Now, consider a defense allocation $d$ and fix an attack $a$ that minimizes
the total attacker cost under $d$.  At most one edge $e \in a$ can have
$d(e) > 0$, for otherwise the cost under $d$ can be reduced by removing an
edge from $a$.  Moreover any attack $a \in \argmin_{e\in\inc(s)} d(e)/w(e)$
minimizes attacker cost under $d$.  Thus the maximin $V$ is witnessed by
defense allocations that maximize $\min_{e \in \inc(s)} d(e)/w(e)$. This
maximization is achieved by allocation $\tilde{d}$ and so
Inequality~(\ref{eq:V-lower}) is an equality.
\end{proof}
We are now ready to prove the main ROA theorem:
\begin{proof}[of Theorem~\ref{thm:ROA-competitive-online}]
First, observe that for all $B>0$ and all $A,C\in\mathbb{R}$
\begin{eqnarray}
\frac{A}{B}\,\leq\, C & \Longleftrightarrow & A-B\,\leq\,(C-1)B\enspace.\label{eq:regret-to-ratio}
\end{eqnarray}
We will use this equivalence to convert the regret bound on profit to the
desired bound on ROA.  Together Theorem~\ref{thm:profit-regret-online} and
Lemma~\ref{lem:max-gt-value} imply
\begin{eqnarray*}
&&\alpha\sum_{t=1}^T\cost(a_t,d_t) \nonumber \\
&\geq& \alpha\max_{d^\star\in\mathcal{D}_{B,E}}
\sum_{t=1}^T\cost(a_t,d^\star)- 
\alpha\frac{B}{2}\sqrt{T\log|E|}-\alpha B\left(\log|E|+\overline{w^{-1}}\right)\nonumber \\
&\geq& \alpha VT-\alpha\frac{B}{2}\sqrt{T\log|E|}
-\alpha B\left(\log|E|+\overline{w^{-1}}\right)\label{eq:sum-cost-lower}
\end{eqnarray*}
where $V=\max_{d\in\mathcal{D}_{B,E}}\min_a\cost(a,d)>0$. If
\begin{gather*}
\sqrt{T} \geq \frac{13}{\sqrt{2}}\left(1+\alpha^{-1}\right)
\sqrt{\log|E|}\sum_{e\in \inc(s)}w(e)\enspace,
\end{gather*}
we can use inequalities $V = B/\sum_{e\in \inc(s)}w(e)$, 
$\overline{w^{-1}}\leq 2\log|E|$ (since $|E|>1$), and
$\left(\sum_{e\in\inc(s)}w(e)\right)^{-1}\leq 1$ to show
\begin{eqnarray*}
\sqrt{T} \geq \left((1+\alpha)B+\sqrt{\left[(1+\alpha)B+24\alpha V\right](1+\alpha)B}\right) (2\sqrt{2}\alpha V)^{-1}\sqrt{\log|E|}\enspace,
\end{eqnarray*}
which combines with Theorem~\ref{thm:profit-regret-online} and
Inequality~\ref{eq:sum-cost-lower} to imply
\begin{eqnarray*}
\alpha\sum_{t=1}^T\cost(a_t,d_t) 
&\geq& \alpha VT-\alpha\frac{B}{2}\sqrt{T\log|E|}
      -\alpha B\left(\log|E|+\overline{w^{-1}}\right) \\
&\geq& \frac{B}{2}\sqrt{T\log|E|}+B\left(\log|E|+\overline{w^{-1}}\right)  \\
&\geq& \sum_{t=1}^T\profit(a_t,d_t)
      -\min_{d^\star\in\mathcal{D}_{B,E}}\sum_{t=1}^T\profit(a_t,d^\star) \\
&=& \sum_{t=1}^T\left(-\cost(a_t,d_t)\right)
      -\min_{d^\star\in\mathcal{D}_{B,E}}
       \sum_{t=1}^T\left(-\cost(a_t,d^\star)\right)  \\
&=& \max_{d^\star\in\mathcal{D}_{B,E}}\sum_{t=1}^T\cost(a_t,d^\star)
      -\sum_{t=1}^T\cost(a_t,d_t)\enspace.
\end{eqnarray*}
Finally, combining this equation with Equivalence~\ref{eq:regret-to-ratio}
yields the result
\begin{eqnarray*}
&&\frac{\ROA\left(\{a_t\}_{t=1}^T,\{d_t\}_{t=1}^T\right)}
 {\min_{d^\star\in\mathcal{D}_{B,E}}\ROA\left(\{a_t\}_{t=1}^T,d^\star\right)} \\
&=& \frac{\sum_{t=1}^T\payoff(a_t,d_t)}
 {\sum_{t=1}^T\cost(a_t,d_t)}\cdot\max_{d^\star\in\mathcal{D}_{B,E}}
 \frac{\sum_{t=1}^T\cost(a_t,d^\star)}{\sum_{t=1}^T\payoff(a_t,d^\star)} \\
&=& \frac{\max_{d^\star\in\mathcal{D}_{B,E}}\sum_{t=1}^T\cost(a_t,d^\star)}
 {\sum_{t=1}^T\cost(a_t,d_t)} \\
&\leq& 1+\alpha\enspace.
\end{eqnarray*}
\end{proof}

% \begin{rem}
% \todo{is this still reasonable?}
% We recognize the convergence time for the competitive ratio as arising from
% elementary electrical physics.  If we imagine that system's graph is an
% electrical circuit with the resistance for wire $e$ given by $1/w(e)$, then
% $\left(\sum_{e\in inc(s)}w(e)\right)^{-1}$ is the total resistance over
% the (parallel) wires to the start vertex $s$.  For the bound on the
% competitive ratio to hold, it is sufficient for $T$ to be at least of the
% order $\left(1+\alpha^{-1}\right)^{2}\log|E|$ divided by this resistance
% squared: the less resistance, the larger $T$ we require for the bound to
% hold.
% 
% The role of this quantity is different to the role of min-cut in the perimeter
% defense setting.  The former arises for an attacker that is purely minimizing
% cost within an arbitrary system, while the latter arises when the attacker
% wishes to maximize ROA/profit in the special case of a unique reward-granting
% node.
% \end{rem}
 
%\presec

\section{Lower Bounds} \label{apx:lower-bound}

We briefly argue the optimality of Algorithm~\ref{alg:mw-online-edges} for a
particular graph, i.e.~we show that Algorithm~\ref{alg:mw-online-edges} has
optimal convergence time for small enough $\alpha$, up to constants.  (For
very large $\alpha$, Algorithm~\ref{alg:mw-online-edges} converges in
constant time, and therefore is optimal up to constants, vacuously.)
The argument considers an attacker who randomly selects an attack path,
rendering knowledge of past attacks useless.  Consider a two-vertex graph
where the start vertex $s$ is connected to a vertex $r$ (with reward $1$) by
two parallel edges $e_1$ and $e_2$, each with an attack surface of $1$.
Further suppose that the defense budget $B = 1$.  We first show a lower
bound on all reactive algorithms:

\begin{lem}
for all reactive algorithms $A$, the competitive
ratio $C$ is at least $(x+ \Omega(\sqrt{T}))/x$, i.e.~at least $(T+
\Omega(\sqrt{T}))/T$ because $x \leq T$.
\end{lem}
\begin{proof}
Consider the following random attack sequence: For each round, select an
attack path uniform IID from the set $\{e_1, e_2\}$. A reactive strategy
must commit to a defense in every round without knowledge of the attack, and
therefore every strategy that expends the entire budget of $1$ inflicts an
expected cost of $1/2$ in every round.  Thus, every reactive strategy
inflicts a total expected cost of (at most) $T/2$, where the expectation is
over the coin-tosses of the random attack process.

Given an attack sequence, however, there exists a proactive defense
allocation with better performance.  We can think of the proactive defender
being prescient as to which edge ($e_1$ or $e_2$) will be attacked most
frequently and allocating the entire defense budget to that edge.  It is
well-known (for instance via an analysis of a one-dimensional random walk)
that in such a random process, one of the edges will occur
$\Omega(\sqrt{T})$ more often than the other, in expectation.

By the probabilistic method, a property that is true in expectation must
hold existentially, and, therefore, for every reactive strategy $A$, there
\emph{exists} an attack sequence such that $A$  has a cost $x$, whereas the
best proactive strategy (in retrospect) has a cost $x+\Omega(\sqrt{T})$.
Because the payoff of each attack is $1$, the total reward in either case is
$T$.  The prescient proactive defender, therefore, has an ROA of $T/(x+
\Omega(\sqrt{T}))$, but the reactive algorithm has an ROA of $T/x$,
establishing the lemma.
\end{proof}

Given this lemma, we show that Algorithm~\ref{alg:mw-online-edges} is optimal given the
information available.  In this case, $n=2$ and, ignoring constants from
Theorem~\ref{thm:ROA-competitive-online}, we are trying to match a
convergence time $T$ is at most $(1+ \alpha^{-1})^2$, which is approximately
$\alpha^{-2}$ for small $\alpha$.  For large enough $T$, there exists a
constant $c$ such that $C \geq (T+ c\sqrt{T})/T$.  By easy algebra, $(T+
c\sqrt{T})/T \geq 1+ \alpha$ whenever $T \leq c^2/\alpha^2$, concluding the
argument.

We can generalize the above argument of optimality to $n > 2$ using the
combinatorial Lemma~3.2.1 from~\cite{expadv}.  Specifically, we can show
that for every $n$, there is an $n$ edge graph for which Algorithm~\ref{alg:mw-online-edges}
is optimal up to constants for small enough $\alpha$.

\end{document}